\newif\ifconf
\conftrue 
\conffalse

\ifconf
  \documentclass{sig-alternate}
  \usepackage{times}
\else
  \documentclass[letterpaper,11pt]{article}

  \usepackage{typearea}
  \typearea{15}
  \usepackage[compact]{titlesec}

\fi

\usepackage{theorem,latexsym,graphicx}
\usepackage{amsmath,amssymb,enumerate}
\usepackage{xspace}

\usepackage{ifpdf}
\usepackage{color}

\allowdisplaybreaks

\ifconf
\newcommand{\lref}[2][]{{#1~\ref{#2}}}
\else
\definecolor{Darkblue}{rgb}{0,0,0.4}
\definecolor{Brown}{cmyk}{0,0.81,1.,0.60}
\definecolor{Purple}{cmyk}{0.45,0.86,0,0}
\newcommand{\mydriver}{hypertex}
\ifpdf
 \renewcommand{\mydriver}{pdftex}
\fi
\usepackage[breaklinks,\mydriver]{hyperref}
\hypersetup{colorlinks=true,
            citebordercolor={.6 .6 .6},linkbordercolor={.6 .6 .6},%
citecolor=Darkblue,urlcolor=black,linkcolor=Darkblue,pagecolor=black}
\newcommand{\lref}[2][]{\hyperref[#2]{#1~\ref*{#2}}}
\fi

\ifconf
\else
\fi

\newtheorem{theorem}{Theorem}[section]

\newtheorem{lemma}[theorem]{Lemma}

\ifconf
\else
\newenvironment{proof}{

\noindent{\bf Proof:}}
{\hfill$\blacksquare$

}
\fi

\newcommand{\junk}[1]{}
\newcommand{\ignore}[1]{}

\newcommand{\Z}[0]{{\ensuremath{\mathbb{Z}}}}

\def\floor#1{\lfloor #1 \rfloor}
\def\ceil#1{\lceil #1 \rceil}

\newcommand{\sse}{\subseteq}
\newcommand{\e}{\varepsilon}
\newcommand{\eps}{\varepsilon}

\newcounter{note}[section]

\newcommand{\qedsymb}{\hfill{\rule{2mm}{2mm}}}
\renewenvironment{proof}{\begin{trivlist} \item[\hspace{\labelsep}{\bf
\noindent Proof.\/}] }{\qedsymb\end{trivlist}}%

\newcommand{\initOneLiners}{%
    \setlength{\itemsep}{0pt}
    \setlength{\parsep }{0pt}
    \setlength{\topsep }{0pt}
}
\newenvironment{OneLiners}[1][\ensuremath{\bullet}]
    {\begin{list}
        {#1}
        {\initOneLiners}}
    {\end{list}}

\DeclareMathOperator*{\mymod}{mod}

\newcommand{\diam}{\mathsf{diam}}
\newcommand{\dAP}{$d$-{\sc dimAP}\xspace}
\newcommand{\dAPplus}{$d$-{\sc dimAP}+\xspace}
\newcommand{\twoAP}{$2$-{\sc dimAP}\xspace}
\newcommand{\twoAPplus}{$2$-{\sc dimAP}+\xspace}
\newcommand{\etal}{\emph{et al.}}

\title{Minimum $d$-dimensional arrangement with fixed points}

\author{
Anupam Gupta
\thanks{Computer Science Department, Carnegie Mellon
    University, Pittsburgh, PA 15213, and Microsoft Research SVC,
    Mountain View, CA 94043. Supported in part by
    NSF awards CCF-0964474 and CCF-1016799, US-Israel BSF grant
    \#2010426, and by the
    CMU-MSR Center for Computational Thinking.}
\and
Anastasios Sidiropoulos
\thanks{University of Illinois at Urbana-Champaign; sidiropo@gmail.com; http://sidiropoulos.org.}
}

\begin{document}

\maketitle
\setcounter{page}{0} 
\thispagestyle{empty}

\begin{abstract}
  In the Minimum $d$-Dimensional Arrangement Problem (\dAP) we are given
  a graph with edge weights, and the goal is to find a $1$-$1$ map of
  the vertices into $\mathbb{Z}^d$ (for some fixed dimension $d\geq 1$)
  minimizing the total weighted stretch of the edges. This problem
  arises in VLSI placement and chip design. 

  Motivated by these applications, we consider a generalization of \dAP,
  where the positions of some of the vertices (pins) is \emph{fixed} and
  specified as part of the input. We are asked to extend this partial
  map to a map of all the vertices, again minimizing the weighted
  stretch of edges. This generalization, which we refer to as \dAPplus,
  arises naturally in these application domains (since it can capture
  blocked-off parts of the board, or the requirement of power-carrying
  pins to be in certain locations, etc.).  Perhaps surprisingly, very
  little is known about this problem from an approximation viewpoint.

  For dimension $d=2$, we obtain an $O(k^{1/2} \cdot \log
  n)$-approximation algorithm, based on a strengthening of the
  \emph{spreading-metric} LP for \twoAP.  The integrality gap for this
  LP is shown to be $\Omega(k^{1/4})$.  We also show that it is NP-hard
  to approximate \twoAPplus~within a factor better than
  $\Omega(k^{1/4-\eps})$.  We also consider a (conceptually harder, but
  practically even more interesting) variant of \twoAPplus, where the
  target space is the grid $\mathbb{Z}_{\sqrt{n}} \times
  \mathbb{Z}_{\sqrt{n}}$, instead of the entire integer lattice
  $\mathbb{Z}^2$.  For this problem, we obtain a $O(k \cdot
  \log^2{n})$-approximation using the same LP relaxation.  We complement
  this upper bound by showing an integrality gap of $\Omega(k^{1/2})$,
  and an $\Omega(k^{1/2-\eps})$-inapproximability result.

  Our results naturally extend to the case of arbitrary fixed target
  dimension $d\geq 1$.
\end{abstract}

\newpage

\section{Introduction}

The Minimum $d$-Dimensional Arrangement Problem (\dAP) is a classical graph embedding, with connections to VLSI layout, multi-processor placement, graph drawing, and visualization (see \cite{Hansen89,LR99,HeldKRV11}).
An input to \dAP~is a graph $G=(E,V)$, with edge weights $\{w_e\}_{e\in E}$, and the goal is to find an injection $f:V\to \{1,\ldots,r\}^d$, for some $d>0$, minimizing the total weighted stretch, i.e.,
\[
\sum_{\{u,v\}\in E} w_{uv} \cdot \|f(u)-f(v)\|_2.
\]
In order to simplify the exposition, we will assume that $n=r^d$ (see \cite{RV} for a detailed discussion).

We consider a natural generalization of \dAP, where the desired mapping $f$ must satisfy certain placement constraints.
More specifically, the positions of some vertices are \emph{fixed}, i.e.,~we are given a set of $k$ \emph{terminals} $T\subset V$, and a mapping $f_0:T \to \mathbb{Z}^d$, and we seek some $f$ that \emph{extends} $f_0$ to $V$, i.e.,~such that $f|_T = f_0$. We refer to this more general problem as \dAPplus.

The above definition models application scenarios where we have some
partial information about the desired embedding.  For example, in the
context of VLSI layout, the locations of some parts of the circuit could
be fixed due to external constraints, e.g., to interface with other
chips, or with the power input; see,~\cite[Section~2.2]{HeldKRV11} for
one ``simplified'' formalization which is even more general than
\dAPplus. However, to the best of our knowledge, the approximability of
this problem has never been studied before --- prior to this work,
nothing was known even for the case of $k=2$ terminals. (The \dAPplus
problem was explicitly posed as an open problem by Jens Vygen.)
Moreover, the setting of fixed pins also captures the setting of
blockages where parts of the grid are forbidden, since we can just
create a set of disjoint vertices that are mapped to those points.

\subsection{Previous work on \dAP}
Hansen \cite{Hansen89}, using the balanced separator result of Leighton and Rao \cite{LR99}, obtained the first $O(\log^2 n)$-approximation algorithm for \dAP.
Combined with the improved separator algorithm of Arora, Rao and Vazirani \cite{ARV}, the algorithm from \cite{Hansen89} yields a $O(\log^{3/2} n)$-approximation.
Subsequently, a $O(\log n \cdot \log \log n)$-approximation was obtained by Even \etal~\cite{ENRS}.
Recently, Rotter and Vygen \cite{RV} have combined the algorithm from \cite{ENRS} with the tree-embedding result of Fakcharoenphol, Rao and Talwar \cite{FRT} to obtain a $O(\log n)$-approximation.
This last algorithm gives the currently best-known approximation guarantee for \dAP.
We remark that Charikar, Makarychev and Makarychev \cite{CMM} have claimed a $O(\sqrt{\log n})$-approximation for \dAP, but a flaw in their argument was subsequently pointed out in \cite{RV}.
The case $d=1$ is also known as Minimum Linear Arrangement (MLA).
We note that better algorithms are known for MLA \cite{CMM,FL,RR}.

There is much practical interest in graph layout problems in VLSI design
with a vast literature. We merely refer to~\cite{HeldKRV11} for a
combinatorial optimization perspective on the problem, and the
references therein.

\subsection{The bounded vs.~the unbounded case}
The currently best-known algorithms for \dAP~give the same approximation guarantee regardless of whether the target space is the integer lattice $\mathbb{Z}^d$, or the $d$-dimensional mesh $\{1,\ldots,n^{1/d}\}^d$.
We refer to the former case as \emph{unbounded}, and to the latter as \emph{bounded} \dAP.
Similarly, one can define the bounded, and unbounded version of \dAPplus.
Even though for \dAP~both bounded, and unbounded variants seem comparable\footnote{We remark that all previous algorithms for approximating \dAP~compute an embedding into the $d$-dimensional grid $\{1,\ldots,n^{1/d}\}^d$. It can be shown by the techniques in \cite{RV,ENRS} that the cost of an optimal arrangement into the $d$-dimensional grid is no more than a $O(\log n)$ factor worse than the cost of an arrangement into the integer lattice.
We are not aware of any better bounds on the ratio of these two quantities.}
 from an approximation point of view, the situation is more intricate for \dAPplus.
More specifically, obtaining an approximation algorithm for bounded \dAPplus~appears to be a more difficult task.
We elaborate on our precise results in the following.

\subsection{Our results}
We consider a natural extension of the \emph{spreading metric} LP (see  \cite{ENRS}) to our problem.
For the case of bounded \dAPplus~we obtain a $O(k^{1/d} \cdot \log n)$-approximation via this LP, where $k$ is the number of terminals.
We also show that the integrality gap in this case is at least $\Omega(k^{1/(2d)})$.
Further, we show that bounded \dAPplus~is NP-hard to approximate within a factor better than $\Omega(k^{1/(2d)-\eps})$, for any fixed $\eps>0$.

These results can also be extended to the case of bounded \dAPplus.
More specifically, using the same LP relaxation, we give a $O(k \cdot \log^2 n)$-approximation for bounded \dAPplus.
The integrality gap for this case is $\Omega(k^{1/d})$, and we can show that the problem is NP-hard to approximate within a factor better than $\Omega(k^{1/d-\eps})$, for any $\eps>0$.

The following table summarizes our results.

\begin{center}
\begin{tabular}{|l||c|c|c|}
\hline
\textbf{Problem} & \textbf{Approximation factor} & \textbf{Integrality gap} & \textbf{Inapproximability}\\
\hline
\hline
Unbounded \dAPplus & $O\left(k^{1/d} \cdot \log n\right)$ & $\Omega\left(k^{1/(2d)}\right)$ & $\Omega\left(k^{1/(2d)-\eps}\right)$\\
\hline
Bounded \dAPplus & $O\left(k \cdot \log^2 n\right)$ & $\Omega\left(k^{1/d}\right)$ & $\Omega\left(k^{1/d-\eps}\right)$\\
\hline
\end{tabular}
\end{center}

Finally, we consider a relaxed version of bounded \dAPplus, where we are allowed to violate the boundaries of the target grid by some small amount.
More specifically, we show that for any fixed dimension $d\geq 1$, and for any $\eps>0$, there exists a polynomial-time $O((k/\eps)^{1/d} \log n)$-approximation algorithm for unbounded \dAPplus, such that the image of the embedding is contained inside the $d$-dimensional grid  $\{1,\ldots,(1+\eps) n^{1/d}\}^d$.

\subsection{Our techniques}
Our algorithm is based on the \emph{spreading metric} LP due to Even \etal~\cite{ENRS}.
A feasible solution to this relaxation is a metric $d$ on $V$.
Intuitively, for an integral solution, the metric $d$ corresponds to the Euclidean metric between the images of the vertices in $V$.
The metric $d$ satisfies certain \emph{spreading} constraints.
Such a constrain states that for a subset $S\subseteq V$, the average distance between points in $S$ is $\Omega(|S|^{1/d})$ (see the next Section for a more precise definition), which holds for any integer solution due to isoperimetric properties of the $d$-dimensional lattice.
The algorithm of \cite{RV} begins by solving this spreading metric LP.
Given an optimal solution $d$, it embeds $d$ into a tree, and then obtains a bijection $b:V\to\{1,\ldots,n\}$ by an appropriate traversal of the tree.
The resulting ordering has small \emph{$d$-dimensional cost}, i.e.,~it satisfies that $\sum_{\{u,v\}\in E} w_{uv} \cdot |b(u)-b(v)|^{1/d}$ is ``small'' compared to the cost of the LP.
One can then obtain an embedding into $\mathbb{Z}^d$ by mapping $V$ along an appropriate iteration of a space-filling curve, in the order given by $b$.

We consider the natural extension of this spreading metric LP in our setting.
It turns out that we can use this LP for both the bounded, and unbounded variants of \dAPplus, yet with different analysis, and guarantees.

Consider now a metric $d$, which is an optimal fractional solution to our LP.
The first obstacle that we face when rounding $d$ is that we cannot map the vertices along a space-filling curve, because such a curve would have to pass through the terminals at specific locations.
We resolve this issue by first using $d$ to \emph{cluster} $V$.
For every terminal $t_i\in T$, we form a cluster $C_i \subseteq V$, containing $t_i$, using the partitioning scheme due to C\v{a}linescu, Karloff and Rabani \cite{CKR}.
The analysis from \cite{CKR} allows us to delete the edges with endpoints in different clusters, and ``charge'' their cost to the remaining edges.

Given the above clustering, it is easy to compute an embedding $f_i$ for each cluster $C_i$ separately using the algorithm from \cite{RV,ENRS}.
The difficult part is how to merge the embeddings $f_1,\ldots,f_k$ of different clusters into a single embedding $f$ for the whole graph.

\paragraph{Merging the clusters in the unbounded case.}
Let us first describe the merging procedure for unbounded \dAP, which turns out to be simpler.
The main difficulty is that if we were to simply ``overlay'' the embeddings $f_1,\ldots,f_d$, different vertices might collide, i.e.,~they can end up being mapped to the same point.
The idea is that since the image of the embedding does not have to be contained inside a bounded area, we can scale every embedding $f_i$ by a factor of $\Theta(k^{1/d})$, increasing the total cost by a factor of $\Theta(k^{1/d})$.
After the scaling, there is enough space to \emph{interleave} the different embeddings, without introducing any collisions.
After the interleaving step, we can easily guarantee, via a small perturbation, that the terminals get mapped to their specified positions.

\paragraph{Merging the clusters in the bounded case.}
The above scaling trick is not applicable to the bounded case, and this
makes the bounded problem a lot more challenging.
The reason is that scaling an embedding $f_i$ can cause the images of some vertices to be mapped outside the grid.
We therefore have to resort to the following more elaborate approach:
First, we partition each cluster $C_i$ into \emph{sub-clusters} $C_{i,0},\ldots,C_{i,\ell}$, for some $\ell\leq \log n$, where each $C_{i,j}$ has size $\Theta(2^j)$, and $C_{i,0}$ contains only the $i$-th terminal $t_i$.
Next, we compute embeddings for the sub-clusters by considering them in increasing order (i.e.,~$C_{1,0},\ldots,C_{k,0},C_{1,1},\ldots$).
When considering a sub-cluster $C_{i,j}$, we can find a rectangle $R$ of constant aspect ratio containing the image of the $i$-th terminal $f_0(t_i)$, with $R$ containing $|C_{i,j}|$ unassigned points, and such that the ratio of unassigned vs assigned points is at least $\Omega(1/k)$.
The final step is to map the vertices in $C_{i,j}$ onto the unassigned points in $R$.
This would have been an easy talk if we had a space-filling curve covering the unassigned portion $R'$ of $R$.
Unfortunately, $R'$ could be very complicated (in fact, $R'$ does not need even be connected).
Despite this, we can find an ordering of the points in $R'$ that, intuitively, at all scales, a typical distance has the same behavior as in a space-filling curve.

\paragraph{Integrality gaps, and inapproximability.}
The main idea for obtaining both of our integrality gaps (both for
bounded, and for unbounded \dAPplus) is constructing a metric space that satisfies all spreading constraints, yet it does not admit a low-stretch embedding into $d$-dimensions:
We color all terminal points blue, and all non-terminal points red.
We place all points with the same color in one plane, and we embed the two planes into $\mathbb{R}^{d+1}$, so that they are parallel, and at distance $r>0$.
Then, we perturb the red points, so that each point moves by at most $O(r)$, and without violating the spreading constraints.
In any low-stretch embedding of the resulting perturbed space into $d$-dimensions, some red point has to collide with some blue point, which leads to a contradiction.
A modification of the above construction allows us to encode a certain type of packing constraints.
Subsequently, we can use these packing constraints to prove hardness of approximation via a reduction from the problem $3$Partition.

\subsection{Organization}
All of the main ideas in both our upper, and lower bounds, are captured in the case of dimension $d=2$.
In order to simplify the exposition, we first present the arguments for \twoAPplus, and subsequently discuss their generalizations to any  $d\geq 1$.
In Sections \ref{sec:algo_unbounded}, and \ref{sec:algo_bounded} we give the algorithms for unbounded, and bounded \twoAPplus~respectively.
In Appendix~\ref{sec:gaps} we present the integrality gaps, and finally in Appendix~\ref{sec:hardness} we prove the hardness of approximation.




\section{An algorithm for unbounded \dAPplus}
\label{sec:algo_unbounded}

We now present the algorithm for the unbounded \twoAPplus problem, and its generalization to \dAPplus, for any $d\geq 1$. As in the case
of algorithms for the minimum linear arrangement, \dAP, and other such problems, we will use a linear program
rounding-based algorithm.

\subsection{The linear program}
\label{sec:lp}

The natural spreading metric linear program is the following:
\begin{align}
  \min \sum_{uv \in E} & w_{uv} \cdot d(u,v) \tag{LP1} \label{eq:1} \\
  \sum_{v \in S} d(u,v) &\geq (|S| - 1)^{1+1/2}/4 \qquad\forall S \sse V, u
  \in S \label{eq:2} \\
  d(u,v) &= \| f_0(u) - f_0(v) \|_2 \qquad\forall u,v \in T \label{eq:3}
  \\
  d & ~~~\text{metric}.
\end{align}
It is well-known that this problem can be solved in polynomial time,
using the ellipsoid method, since the separation oracle for the
constraints~(\ref{eq:2}) can be implemented using a shortest-path
algorithm (see, e.g.,~\cite{LR99,ENRS,RV}). One issue is that the
distance $\| f_0(u) - f_0(v) \|_2$ may be irrational; let us assume that
in this case we replace the constraint~(\ref{eq:3}) by
\[ d(u,v) \in (1 \pm \varepsilon) \| f_0(u) - f_0(v) \|_2 \qquad\forall
u,v \in T, \] where $\varepsilon > 0$ is some tiny constant. This is
clearly a relaxation, and it can be verified that it does not change the
subsequent rounding arguments. Let $d$ be the optimal metric returned,
with LP value $Z^*$.

\subsection{Rounding}
\label{sec:rounding}

To round the linear program solution, we will use algorithms for the
$0$-Extension problem. These algorithms will partition the vertex set
into $k$ ``clusters'', each containing a single terminal. We then use an
existing algorithm of Rotter and Vygen~\cite{RV} to lay each of these
clusters out in the plane. Since the resulting map might not be injective,
we then use an interleaving technique to give us a legal map.

\subsubsection{Partitioning using $0$-Extension} 
\label{sec:partitioning}

\newcommand{\C}{\mathcal{C}}

First, we use a $0$-Extension algorithm to assign each point to a
terminal. In the $0$-Extension we are given a graph $G = (V,E)$ with the
edges having weights $w_e$, and a set of terminals $T \sse V$, and a
metric $d(\cdot, \cdot)$ on the terminal set. The goal is to find a
``retraction'' map $g: V \to T$ such that $g(t) = t$ for all terminals
$t \in T$, to minimize $\sum_{uv \in E} w_{uv}\, d(g(u), g(v))$. One
technique (rounding the semimetric relaxation) developed for this
problem can be abstracted out as follows:

\begin{theorem}[\cite{CKR}]
  \label{thm:ckr}
  Given a metric $(V,d)$ on the vertices of the graph $G = (V,E)$,
  and a subset $T \sse V$, there exists a
  randomized poly-time algorithm that outputs a map $g: V \to T$ such
  that 
  \begin{OneLiners}
  \item[(a)] $g(t) = t$ for all $t \in T$,
  \item[(b)] if $A_u^* := d(u, T)$ is the  $u$'s distance to its closest
    terminal, then $\Pr[ d(u, g(u)) \leq 2A_u^*] = 1$, and
  \item[(c)] The probability that $u$ and $v$ are mapped to different
    terminals is 
    \[ \Pr[ g(u) \neq g(v) ] \leq \alpha_{CKR} \cdot \frac{ d(u,v)
    }{(A_u^* + A_v^*)}. \]
  \end{OneLiners}
  where $\alpha_{CKR} = O(\log k)$.\footnote{The upper bound shown in
    the paper was technically $\alpha_{CKR} \cdot \frac{ d(u,v)
    }{\min(A_u^*,A_v^*)}$, but the two bounds are easily shown to be
    equivalent within constant factors.}
\end{theorem}
We apply Theorem~\ref{thm:ckr} to the graph $G$ with the metric $d$
given by the linear program~(\ref{eq:1}) to get the random retraction
map $g$.  Such a map $g$ can also be viewed as a clustering of the
vertices into $k$ clusters $C_1, C_2, \ldots, C_k$, where $C_i =
g^{-1}(t_i)$ for each $i \in \{1,2,\ldots,k\}$. (Here we denote the
terminals in $T$ by $\{t_1,t_2,\ldots,t_k\}$.) Use $\C_g = \{C_1, C_2,
\ldots, C_k\}$ to denote the random clustering.


\subsubsection{Altering the instance}
\label{sec:altering}

Based on the clustering $\C_g$, we now alter the \dAPplus~instance $G =
(V,E)$ to a slightly different instance $H = (V,E_H)$. For each edge $uv
\in E$ that goes between two clusters (say $u \in C_{1}$ and $v \in
C_2$), remove the edge $(u,v)$ and add in the edges $(u,t_1)$, $(t_1,
t_2)$ and $(v,t_2)$, all of the same weight $w_{uv}$. Note that we have
added two intra-cluster edges, and one terminal-terminal edge. This
being done, all remaining edges in $H$ are either (a)~intra-cluster
edges, or (b)~terminal-terminal edges. Note that we may have parallel
edges in $H$.

The original LP solution $d$ remains a feasible solution to~(\ref{eq:1})
on this new instance $H$, with a potentially higher cost. What is this
new cost? For each edge $uv \in E$, if it is not cut by the clustering,
we pay the same amount on $H$ as on $G$. If $uv$ is cut (say $g(u) =
t_1$ and $g(v) = t_2$), we now pay:
\[ w_{uv} \cdot \big( d(u, t_1) + d(t_1,t_2) + d(t_2, v) \big) \leq
w_{uv} \cdot \big(A_u^* + (A_u^* + A_v^*+ d(u,v)) + A_v^* \big) \] where we
used the triangle inequality and Theorem~\ref{thm:ckr}(b). Hence the
expected \emph{increase} in the LP cost paid by $uv$ is $4w_{uv} (A_u +
A_v)$ times the probability that $u,v$ are separated, which using
Theorem~\ref{thm:ckr}(c) is at most
\[ 4w_{uv}(A_u + A_v) \cdot \alpha\, \frac{d(u,v)}{A_u + A_v} =
4\alpha\, w_{uv}\, d(u,v). \] Summing over all edges, the expected cost
of~(\ref{eq:1}) on the new instance $H$ is at most $Z' \leq (4\alpha +
1)Z^*$.

Now suppose we can find a layout $\widehat{f}: V \to \Z^2$ for the new
instance $H$ with cost
\[ \sum_{xy \in E_H} \| \widehat{f}(x) - \widehat{f}(y) \| \leq \beta
Z', \] (i.e., a $\beta$-approximation for the new instance with respect
to the LP). We claim that $\widehat{f}$ is a good layout of the original
instance $G$ as well. Indeed, we know that for any edge $uv \in E$ that
was transformed into three edges,
\[ \| \widehat{f}(u) - \widehat{f}(v) \| \leq \| \widehat{f}(u) -
\widehat{f}(t_1) \| + \| \widehat{f}(t_1) - \widehat{f}(t_2) \| + \|
\widehat{f}(t_2) - \widehat{f}(v) \|. \] And all three terms have been
paid for by the three edges in the new instance, hence the cost of the
$\beta$-approximate layout $\widehat{f}$ on the original instance $G$ is
at most $\beta Z' \leq O(\beta \cdot \alpha \cdot Z^*)$. Hence, it
suffices to find a good layout for an instance where all edges are
either intra-cluster edges or terminal-terminal edges.

\subsubsection{Solving the new instance} Use the Rotter-Vygen algorithm
independently on each cluster $C_i$ (with only the intra-cluster edges),
to map the points in $C_i$ to the integer mesh, ensuring that the
terminal $t_i$ is mapped to the location $f_0(t_i)$. 
This gives a $2$-dimensional layout $\widehat{f}_i$ for each cluster,
whose cost is at most $\beta = O(\log n)$ times the LP solution for the
inter-cluster edges. (The map $\widehat{f}$ can be thought of using
$\widehat{f}_i$ to independently map each cluster.) Note that this map
$\widehat{f}$ has the same cost as the sum of the costs of the
$\widehat{f}_i$ maps, but is not ``legal'' since we have just mapped
each of the the clusters independently --- we will fix this next. Since
the terminals are going to be fixed by the new layout, we need not worry
about the terminal-terminal edges---we lose nothing and achieve a factor
$\beta = 1$ on these edges.

\subsubsection{Legalization using interleaving}
\label{sec:interleaving}

Let $\ell \in \Z$ such that $(\ell-1)^2 < k \leq \ell^2$. Color the
integer lattice using $\ell^2$ colors, so that points of each color form
a dilated copy of the lattice and are at horizontal and vertical
distance $\ell$ from each other.  Specifically, identify the $\ell^2$
colors with tuples $(i,j) \in [\ell]\times[\ell]$ and color the point
$(x,y) \in \Z^2$ with the color $(x \bmod \ell, y \bmod \ell)$. Now for
each $i$, focus on the sublattice consisting of the points with color
$i$ and use the mapping $\widehat{f}_i$ to map the points of the cluster
$C(t_i)$ onto this sublattice, translating it suitably to map $t_i$ to a
point in this sublattice closest to $f_0(t_i)$. Note that this process
causes the stretch to increase by a factor of $O(\sqrt{k})$ over that in
the mappings $\widehat{f}_i$.

Finally, we may not have ensured that $t_i$ is mapped to the lattice
point $f_0(t_i)$; it may merely be mapped to the $i$-colored point
closest to $f_0(t_i)$. This is easily fixed ``locally''. Indeed, for
some $[\ell]\times[\ell]$ subgrid containing some terminal locations
$f_0(t_i)$, remove the current colors, first color each such point
$f_0(t_i)$ with color $i$, then put the remaining colors down in some
arbitrary way. This local remapping causes each point to move at most
distance $O(\ell)$ from its original position, and causes only an
additional $O(\sqrt{k})$ additive stretch for each edge. 

Overall, the total loss is a factor of $O(\alpha) \times O(\log n)
\times O(\sqrt{k})$ in this approach. The Calinescu et al.~\cite{CKR}
algorithm gives us $\alpha = O(\log k)$ for the above, giving a final
approximation guarantee of $O(\sqrt{k} \log k \log n)$. In
Section~\ref{sec:gaps} we show that~(\ref{eq:1}) has an integrality
gap of $\Omega(k^{1/4})$, and it is NP-hard to approximate within a factor better than $\Omega(k^{1/4-\eps})$, for any fixed $\eps>0$.
This means that while we might be able to
improve on the approximation guarantee using a better rounding, removing
the polynomial dependence in $k$ is impossible, unless P$=$NP.

\subsection{A slight improvement}
\label{sec:slight}

It is possible to remove the $O(\log k)$ term from the approximation
guarantee above by avoiding some redundancy in the $0$-Extension and the
Rotter-Vygen algorithm. For this, we observe that the Rotter-Vygen
algorithm first embeds the metric $d$ into a hierarchically
well-separated tree (HST) using the algorithm of Fakcharoenphol et
al.~\cite{FRT}, which stretches distances by an expected $O(\log n)$
factor. It then uses an inorder traversal of this tree to get an
ordering of the points, and lays them out on the grid using a
space-filling curve. And even before that, our algorithm loses an
$O(\log k)$ while performing the $0$-Extension-based clustering. The
plan to improve the approximation is to reverse the order in which we
perform the clustering and HST-building operations.

Indeed, first we take the entire metric and embed it into an HST $F$
using~\cite{FRT}, so that no distance is shrunk, and each distance is
stretched by $O(\log n)$ in expectation. Denote the distances in the
tree $F$ by $d_F(\cdot, \cdot)$. Now we use an algorithm of Englert et
al.~\cite{EGKRTT} for ``connected $0$-Extension'' to find a random map
$g: V \to T$ such that
\begin{OneLiners}
  \item[(i)] $g(t_i) = t_i$ for all $t_i \in T$,
  \item[(ii)] the nodes in $C_i := g^{-1}(t_i)$ form a connected portion of
    the tree $F$, and 
  \item[(iii)] $\mathbf{E}[ \mathbf{1}_{g(u) \neq g(v)} \times \big( d_F(u, g(u)) + d_F(g(u),
    g(v)) + d_F(v, g(v)) \big) ] = O(1)\cdot d_F(u,v).$
\end{OneLiners}
For property~(iii), the result follows because tree metrics are
$O(1)$-decomposable. (The~\cite{EGKRTT} paper only claims a bound on
$\mathbf{E}[d_F(g(u), g(v))]$, but the bound in property~(iii) follows
easily from their techniques.) Using this to alter the instance as in
Section~\ref{sec:altering}, we get a new instance $H$ with only
intra-cluster and terminal-terminal edges, where the expression~(iii)
ensures that the expected increase in the LP objective function is at
most $O(1) \sum_{uv \in E} w_{uv}\, d_F(u,v)$. Taking expectations over
the choice of the HST $F$, this is at most $O(\log n) \sum_{uv \in E}
w_{uv}\, d(u,v) = O(\log n) \cdot Z^*$.

Finally, the last observation is that each of the clusters $C_i$ already
forms a (connected) HST. Hence we can perform the inorder traversal (as
in Rotter-Vygen) directly on this HST, and avoid any further loss in the
approximation guarantee. Incorporating the loss of $O(\sqrt{k})$ due to
the legalization, this gives us an approximation guarantee of
$O(\sqrt{k} \log n)$. In the next Section, we build on the techniques of
this Section to give an algorithm for the \emph{bounded} case of the
problem, where the vertices must be mapped into a square $n$-vertex
grid.

Finally, the above algorithm immediately extends to the general \dAPplus, for any fixed $d\geq 1$, giving the following result.

\begin{theorem}
For any fixed $d\geq 1$, there exists a polynomial-time, $O(k^{1/d} \cdot \log n)$-approximation algorithm for unbounded \dAPplus.
\end{theorem}


\section{An algorithm for bounded \dAPplus}
\label{sec:algo_bounded}

We now describe an algorithm for the bounded \dAPplus.
Again, in order to simplify the exposition, we first focus on the case of dimension $d=2$.
Let $G=(V,E)$ be an instance to \dAPplus, with $|V|=n$.
We will assume for the sake of notation that $\sqrt{n}$ is an integer, however the same algorithm and analysis work for arbitrary $n$.
We describe the algorithm as a sequence of several main steps:

\paragraph{Step 1: Computing a CKR partition.}
We begin by first computing an optimal fractional solution $d$ to the spreading LP.
As before, we also compute a random partition of $V$ into clusters $C_1,\ldots,C_k$, using the algorithm of Calinescu, Karloff, and Rabani \cite{CKR}.
Moreover, for every edge $\{u,v\} \in E$, with $u\in C(t_i)$, $v\in C(t_j)$, with $i\neq j$, we remove $\{u,v\}$, and we add the edges $\{u,t_1\}$, $\{t_1,t_2\}$, and $\{v,t_2\}$.
Let $H$ be the resulting graph.
Arguing as before, we have that $d$ is a feasible solution for the LP corresponding to this new instance, and the optimal arrangement for $H$ is at most  $O(\log k)$ times the cost of the optimal arrangement for $G$.
It therefore suffices to find a near-optimal solution for $H$.

\paragraph{Step 2: Ordering the vertices of each cluster.}


We will use the following result of Rotter and Vygen \cite{RV}, which follows by the tree-embedding result of Fakcharoenphol, Rao and Talwar \cite{FRT}.

\begin{lemma}[Rotter and Vygen \cite{RV}]\label{lem:ordering}
Let $d$ be a feasible solution for the spreading LP for an instance $G=(V,E)$, and let $v\in V$.
Then, there exists a bijection $q:V \to \{0,\ldots,|V|-1\}$, satisfying the following conditions:
\begin{description}
\item{(1)}
$q(v)=0$.
\item{(2)}
$\sum_{\{u,v\}\in E} w_{uv} \cdot  \sqrt{|q(u)-q(v)|} = O(\log |V|) \sum_{\{u,v\}\in E} w_{uv} \cdot d(u,v)$.
\end{description}
Moreover, $q$ can be computed in polynomial time.
\end{lemma}

For any $i\in \{1,\ldots,k\}$, let $H_i=(C_i, E_i)$ be the subgraph of $H$ induced on $C_i$.
Note that for any $i\in \{1,\ldots,k\}$, the restriction of $d$ on $H_i$ is a feasible solution for the LP for $H_i$.
For every $i\in \{1,\ldots,k\}$, let $q_i:C_i \to \{0,\ldots,|C_i|-1\}$ be a bijection given by Lemma~\ref{lem:ordering}, with $q_i(t_i)=0$.
We have 
\begin{align}
\sum_{\{u,v\}\in E_i} w_{uv} \cdot \sqrt{|q_i(u)-q_i(v)|} = O(\log n) \sum_{\{u,v\} \in E_i} w_{uv} \cdot d(u,v).
\end{align}

\paragraph{Step 3: Partitioning into sub-clusters.}

For any $i\in \{1,\ldots,k\}$, we further compute a random partition ${\cal C}_i=\{C_{i,0},\ldots,C_{i,\delta_i}\}$ of  $C_i$ as follows.
Pick $\gamma \in [1,2)$, uniformly at random.
Let $\delta_i = \max\{0, \lfloor \log (C_i / 4) \rfloor\}$.
For any $j\in \{0,\ldots,\delta_i-1\}$, let 
\[
C_{i,j} = q_i^{-1}\left(\{\lfloor \gamma \cdot 2^j \rfloor, \ldots, \lfloor \gamma \cdot 2^{j+1} \rfloor - 1\}\right),
\]
and
\[
C_{i,\delta_i} = C_i \setminus \left(\bigcup_{j=0}^{\delta_i} C_{i,j}\right)
\]
For a vertex $x\in C_i$, let ${\cal C}_i(x)$ denote the cluster $C_{i,j}\in {\cal C}_i$ containing $x$.
We refer to each $C_{i,j}$ as a \emph{sub-cluster}.
Let $q_{i,j}:C_{i,j} \to \{0,\ldots,|C_{i,j}|-1\}$ be the bijection obtained by restricting the ordering induced by $q_i$ on $C_{i,j}$.

\paragraph{Step 4: Partitioning the grid into regions.}

We next define a partition of $\{1,\ldots,\sqrt{n}\}^2$ into pairwise disjoint subsets $\{Y_{i,j}\}_{i,j}$, such that for any $i\in \{1,\ldots,k\}$, and for any $i\in \{0,\ldots,\delta_i\}$, we have $|Y_{i,j}| = |C_{i,j}|$.
We refer to each $Y_{i,j}$ as a \emph{region}.
The final embedding $f$ will map $C_{i,j}$ onto $Y_{i,j}$, i.e.,~it will satisfy $f(C_{i,j}) = Y_{i,j}$.
We begin by setting for any $i\in \{1,\ldots,k\}$, 
\[
Y_{i,0} = \{f_0(t_i)\}.
\]
For any $i\in \{1,\ldots,k\}$, and for any $j\in \{0,\ldots,\delta_i\}$, we set $Y_{i,j}$ to be a subset of the currently unassigned points in $\{1,\ldots,\sqrt{n}\}^2$, so that the maximum distance to $f_0(t_i)$ is minimized.
More concretely, we set $Y_{i,j}$ to be a subset of 
\[
Z_{i,j} = \{1,\ldots,\sqrt{n}\}^2 \setminus \left( \left( \bigcup_{i'=1}^k \bigcup_{j'=0}^{j-1} Y_{i',j'} \right) \cup \left( \bigcup_{i'=1}^{i-1} Y_{i',j} \right) \right),
\]
minimizing $\max_{x\in Y_{i,j}} \|f_0(t_i) - x\|_{\infty}$, and breaking ties arbitrarily,
where we have used the convention $Y_{i',j'}=\emptyset$ for any $j'>\delta_i$.

\paragraph{Step 5: Covering each region with a pseudo-space-filling curve.}

We use the algorithm of the following lemma to compute an ordering of the points in every $Y_{i,j}$, that resembles the ordering obtained by a space-filling curve.
Note however that each $Y_{i,j}$ might be very complicated, so we have to use a fairly more elaborate construction, which is summarized in the following lemma.

\begin{lemma}\label{lem:Rij_curve}
For any $i\in \{1,\ldots,k\}$, and for any $j\in \{0,\ldots,\delta_i\}$, there exists a bijection $b_{i,j}:Y_{i,j}\to \{0,\ldots,|Y_{i,j}|-1\}$, such that for any $\Delta \in \{0,\ldots,|Y_{i,j}|-1\}$, we have
\[
\sum_{\ell=0}^{|Y_{i,j}|-1-\Delta} \|b_{i,j}^{-1}(\ell)-b_{i,j}^{-1}(\ell+\Delta)\|_2 = O\left(k\cdot |Y_{i,j}| \cdot \sqrt{\Delta} \right).
\]
Moreover, $b_{i,j}$ can be computed in polynomial time.
\end{lemma}

The proof of Lemma~\ref{lem:Rij_curve} is deferred to Section \ref{sec:curve}.

\paragraph{Step 6: Randomly shifting every sub-cluster.}

Next, for every $i\in \{1,\ldots,k\}$, and for every $j\in \{0,\ldots,\delta_i\}$, we pick $\tau_{i,j} \in \{0,\ldots, |C_{i,j}|\}$, uniformly at random.
We define a new bijection $\widetilde{q_{i,j}} : C_{i,j} \to \{0,\ldots,|C_{i,j}|-1\}$ as follows.
For any $x\in C_{i,j}$, we set
\[
\widetilde{q_{i,j}}(p) = q_{i,j}(p) + \tau_{i,j} \mymod |C_{i,j}|,
\]
effectively taking a cyclic shift of the ordering induced by $q_{i,j}$, where every element is translated by $\tau_{i,j}$.

\paragraph{Step 7: The final embedding.}
We are now ready to define the final embedding $f:V \to \{1,\ldots,\sqrt{n}\}^2$.
For every $i \in \{1,\ldots,k\}$, for every $j\in \{0,\ldots, \delta_i\}$, and for every $v\in R_{i,j}$, we set 
\[
f(v) = b_{i,j}^{-1}(\widetilde{q_{i,j}}(v)).
\]
This completes the description of the algorithm.

\subsection{Analysis}

\begin{lemma}\label{lem:bounded_Ei}
For any $i\in \{1,\ldots,k\}$, we have
\[
\mathbf{E}\left[ \sum_{\{u,v\} \in E_i} w_{uv} \cdot \|f(u)-f(v)\|_2 \right] = O(k\cdot \log n) \cdot \sum_{\{u,v\} \in E_i} w_{uv} \cdot d(u,v),
\]
where the expectation is over the random choices in the construction of $f$.
\end{lemma}

\begin{proof}
Let $i\in \{1,\ldots,k\}$, and $u,v\in C_i$.
Assume w.l.o.g.~that $q_i(v) = q_i(u) + \Delta$, for some $\Delta>0$.
Let ${\cal E}$ denote the event that ${\cal C}_i(u) \neq {\cal C}_i(v)$.
We have that 
\begin{align*}
\Pr[{\cal E}] &= O\left( \frac{q_i(v)-q_i(u)}{q_i(v)} \right) = O\left( \frac{\Delta}{q_i(v)} \right).
\end{align*}
Conditioned on the event ${\cal E}$, we have 
\begin{align*}
\|f(u)-f(v)\|_2 &= O\left(\sqrt{ k \cdot q_i(v) }\right).
\end{align*}

On the other hand, suppose that the event ${\cal E}$ does not occur, i.e.,~${\cal C}_i(u) = {\cal C}_i(v) = C_{i,j}$, for some $j$.
Let ${\cal E}'$ be the event that $\widetilde{q_{i,j}}(v) \neq \widetilde{q_{i,j}}(v) + \Delta$.
We have
\begin{align*}
\Pr[{\cal E}'] &= O\left(\frac{q_i(v)-q_i(u)}{q_i(v)}\right) = O\left( \frac{\Delta}{q_i(v)} \right).
\end{align*}
Conditioned on ${\cal E}'$, 
\begin{align*}
\|f(u)-f(v)\|_2 &= O\left( \diam(Y_{i,j}) \right) = O\left(\sqrt{ k \cdot q_i(v) }\right),
\end{align*}
where $\diam$ denotes Euclidean diameter.
On the other hand, by Lemma~\ref{lem:Rij_curve} we have
\begin{align*}
\mathbf{E}\left[ \|f(u)-f(v)\|_2 | {\cal E}' \right] &= O\left(k \cdot \sqrt{\Delta}\right).
\end{align*}
Putting everything together, we have
\begin{align*}
\mathbf{E}\left[ \|f(u) - f(v)\|_2 \right] &= \mathbf{E}\left[ \|f(u) - f(v)\|_2 | {\cal E} \right] \cdot \Pr[{\cal E}] + \mathbf{E}\left[ \|f(u) - f(v)\|_2 | \lnot {\cal E} \right] \cdot \Pr[\lnot {\cal E}]\\
 &= O\left(\sqrt{k} \cdot \frac{\Delta}{\sqrt{q_i(v)}}\right) + O\left( k \cdot \sqrt{\Delta} \right) \\
 &= O\left(\sqrt{k} \cdot \frac{\Delta}{\sqrt{\Delta}}\right) + O\left( k\cdot \sqrt{\Delta} \right) \\
 &= O\left(k \cdot \sqrt{\Delta}\right) \\
 &= O\left(k \cdot \sqrt{|q_i(v) - q_i(u) |}\right) 
\end{align*}

By Lemma~\ref{lem:Rij_curve} and the linearity of expectation, we conclude that
\begin{align*}
\mathbf{E}\left[ \sum_{\{u,v\} \in E_i} w_{uv} \cdot \|f(u)-f(v)\|_2 \right] &= \sum_{\{u,v\} \in E_i} w_{uv} \cdot \mathbf{E}\left[  \|f(u)-f(v)\|_2 \right] \\
 &= O\left( \sum_{\{u,v\} \in E_i} w_{uv} \cdot k \cdot \sqrt{|q_i(u) - q_i(v)|} \right) \\
 &= O(k \cdot \log n) \sum_{\{u,v\} \in E_i} w_{uv} \cdot d(u,v).
\end{align*}
as required.
\end{proof}

We are now ready to prove the main result of this Section.

\begin{theorem}\label{thm:2d-bounded}
There exists a polynomial-time $O(k \cdot \log^2 n)$-approximation for bounded \twoAPplus.
\end{theorem}
\begin{proof}
As in the unbounded case, we pay a $O(\log n)$ factor for the initially clustering step using the partitioning scheme of \cite{CKR}.
Applying Lemma~\ref{lem:bounded_Ei} on each cluster $C_i$, and summing over all $i$, we conclude that the computed embedding is a $O(k \cdot \log n)$-approximation for the modified instance $H$, and therefore  it is a $O(k\cdot \log^2 n)$-approximation for the original instance $G$.
\end{proof}

Finally, the same algorithm directly generalizes to the case of arbitrary dimension $d\geq 1$.

\begin{theorem}\label{thm:bounded}
For any fixed $d\geq 1$, there exists a polynomial-time $O(k \cdot \log^2 n)$-approximation for bounded \dAPplus.
\end{theorem}

We remark that the approximation ratio in Theorem \ref{thm:bounded} does not improve with the dimension, as in the unbounded case.
This is due to the lack of space-filling curves for the sets considered by the algorithm, which forces us to use the slightly weaker bounds given by Lemma~\ref{lem:Rij_curve}.
Our proof of Lemma~\ref{lem:Rij_curve} immediately generalizes to any fixed dimension $d\geq 1$, but the resulting bound does not improve for larger $d$.
Improving Lemma~\ref{lem:Rij_curve} for higher values of the dimension $d$ would directly improve our approximation ratio for \dAPplus.
We leave this as an interesting open problem.

\subsection{Finding a good ordering}
\label{sec:curve}

\begin{lemma}\label{lem:curve_general}
Let $h \geq 0$ be an integer,
and let $X = \mathbb{Z}_{2^h}^2$.
Let $Y\subseteq X$, with $|Y|=\eps\cdot |X|$, for some $\eps > 0$.
Then, there exists a bijection $b:Y \to \{0,\ldots,|Y|-1\}$, such that for any $\Delta \geq 1$, 
\[
\sum_{\ell=0}^{|Y|-\Delta-1} \|b^{-1}(\ell) - b^{-1}(\ell + \Delta)\|_2 = O\left(\frac{\sqrt{\Delta}}{\eps} \cdot |Y|\right).
\]
Moreover, given $X$, the bijection $b$ can be computed in polynomial time.
\end{lemma}
\begin{proof}
Let $T=(V,E)$ be the standard quad-tree on $X$.
Formally, every $v\in V$ corresponds to a square $S_v\subseteq X$.
We consider $T$ as being rooted at a vertex $r$, with $S_r=X$.
Every vertex $v\in T$ has four children $v_1,\ldots,v_4\in V$, such that $S_v = \bigcup_{i=1}^4 S_{v_i}$, and such that each $S_{v_i}$ is a square of side length half the side length of $S_v$.
For every leaf $v\in T$, the square $S_v$ contains a single point from $X$.
For simplicity of notation, we will identify $v$ with the single point in $S_v$.

Let $b$ be defined by the ordering of $Y$ obtained by an in-order traversal of $T$.
More precisely, for a leaf $v\in Y$,  we set $b(v)=i-1$, if $v$ is the $i$-th leaf in $Y$ visited in the in-order traversal of $T$.

Let $i\geq 0$, and let $v\in V$ be a vertex at depth $i$, where we consider the root as being at depth $0$.
We have $|S_v| = |X|/4^i$.
Let $S_v' \subset S_v$ contain all but the $\Delta$ points in $S_v$, that are visited last in the in-order traversal of $T$.
For any $u\in S_v'$, we have that the leaf of $T$ that is visited $\Delta$ steps after $u$ in the in-order traversal, is also in $S_v$, and therefore $\|u-b(b^{-1}(u)+\Delta)\|_2 < \sqrt{2} \cdot 2^h \cdot 2^{-i}$.

It follows that for any $i\in \{0,\ldots,h-\lfloor \log \sqrt{\Delta}\rfloor\}$, there are at most $4^i \cdot \Delta$ points $w\in Y$, with $\|w-b(b^{-1}(w)+\Delta)\|_2 \geq \sqrt{2} \cdot 2^h \cdot 2^{-i}$.
For all remaining points $w\in Y$, we have $\|w-b(b^{-1}(w)+\Delta)\|_2 \leq \sqrt{\Delta}$.
Therefore,
\begin{align*}
\sum_{\ell=0}^{|Y|-\Delta-1} \|b^{-1}(\ell) - b^{-1}(\ell + \Delta)\|_2 &= O(|Y| \cdot \sqrt{\Delta}) + \sum_{w\in Y : b(w) \leq |X|-\Delta-1} \|w - b(b^{-1}(w) + \Delta)\|_2\\
 &= O(|Y| \cdot \sqrt{\Delta}) + \sum_{i=0}^{h-\lfloor \log\sqrt{\Delta}\rfloor} \Delta \cdot 4^i \cdot \sqrt{2} \cdot 2^{h-i} \\
 &= O(|Y| \cdot \sqrt{\Delta}) + O(|X| \cdot \sqrt{\Delta})\\
 &= O(|X| \cdot \sqrt{\Delta})\\
 &= O(|Y| \cdot \sqrt{\Delta} / \eps),
\end{align*}
as required.
\end{proof}

\begin{proof}[Proof of lemma \ref{lem:Rij_curve}]
By the choice of $Y_{i,j}$, there exists a square $X \subseteq \mathbb{Z}^2$, such that $|X|/|Y_{i,j}| = O(k)$.
The assertion follows by setting $b:Y_{i,j} \to \{0,\ldots,|Y_{i,j}|-1\}$ to be the bijection given by Lemma~\ref{lem:curve_general}, for $Y=Y_{i,j}$.
\end{proof}



\subsection*{Acknowledgments}

We thank Bruce Shepherd and Fritz Eisenbrand for their generous
hospitality, and Jens Vygen for kindly posing this problem.

{\small
\bibliographystyle{alpha}
\bibliography{bibfile}
}

\appendix

\section{Integrality gaps}
\label{sec:gaps}

\subsection{Bounded \dAPplus}

\begin{theorem}\label{thm:gap_bounded}
The integrality gap of the spreading LP for the bounded version of \twoAPplus~is $\Omega(\sqrt{k})$.
\end{theorem}

\begin{proof}
Let $G=(V,E)$ be the complete graph on $n$ vertices.
We assume w.l.o.g.~that $\sqrt{n}$ is an integer.
Let $T\subset V$ be the set of terminals, with $|T|=k=n-2$.
Let 
\[
A = (\mathbb{Z}_{\sqrt{n}} \times \mathbb{Z}_{\sqrt{n}}) \setminus \{(0,0), (\sqrt{n}-1, \sqrt{n}-1)\}.
\]
We set the function $f_0:T \to \mathbb{Z}_{\sqrt{n}} \times \mathbb{Z}_{\sqrt{n}}$ to be an arbitrary bijection from $T$ onto $A$.

Let $\{u,v\} = V \setminus T$ be the non-terminal vertices.
We define the weight $w_e$ of any edge $e\in E$ to be
\[
w_e = \left\{\begin{array}{ll}
1 & \text{ if } e = \{u,v\}\\
0 & \text{ otherwise}
\end{array}\right.
\]

We now define a metric $d$ on $V$, that is a feasible solution to the LP.
For any $x,y\in V$, we set
\[
d(x,y) = \left\{\begin{array}{ll}
\|f_0(x) - f_0(y)\|_2 & \text{ if } x,y\in T \\
\sqrt{2 n} & \text{ if } |\{x,y\} \cap  T| = 1\\
1 & \text{ if } x,y \in V \setminus T
\end{array}\right.
\]
It is immediate to check that $d$ defines a metric (i.e.,~it is symmetric, and satisfies the triangle inequality), and that it satisfies all spreading constraints.
Therefore, $d$ is a feasible solution to the LP.
The cost of $d$ is 
\begin{align*}
c_{\mathsf{LP}} &= \sum_{\{x,y\} \in E} w_{xy} \cdot d(x,y) = w_{uv} \cdot d(u,v) = 1.
\end{align*}

On the other hand, any integral solution has to map $u$, and $v$ into the points $(0,0)$, and $(\sqrt{n} - 1, \sqrt{n} - 1)$, since these are the only remaining empty slots in the grid.
Therefore, any integral solution $f$ satisfies $\|f(u)-f(v)\|_2 = \sqrt{2 n}$.
This implies that the cost of $f$ is
\begin{align*}
c_f &= \sum_{\{x,y\} \in E} w_{xy} \|f(x) - f(y)\|_2 = \sqrt{2n}.
\end{align*}
Therefore, the integrality gap is at least $c_f / c_{\mathsf{LP}} = \Omega(\sqrt{k})$, as required.
\end{proof}

The above construction immediately generalizes to arbitrary fixed dimension $d\geq 1$, giving the following result.

\begin{theorem}
For any fixed $d\geq 1$, the integrality gap of the spreading LP for the bounded version of \dAPplus~is $\Omega(k^{1/d})$.
\end{theorem}

\subsection{Unbounded \dAPplus}

\begin{theorem}\label{thm:gap_bounded}
The integrality gap of the spreading LP for the unbounded version of \twoAPplus~is $\Omega(k^{1/4})$.
\end{theorem}

\begin{proof}
Let $G=(V,E)$ be the complete graph on $9n$ vertices.
Let $t>0$ be a positive integer parameter to be specified later.
We assume w.l.o.g.~that $\sqrt{n}$ is an integer multiple of $2t$.
Let $T\subset V$, be the set of terminals, with $|T| = k = 9n - n/t^2$.
Let
\[
B = t \cdot \mathbb{Z}^2 \cap \{\sqrt{n}, \ldots, 2\sqrt{n} - 1\}^2,
\]
i.e.,~$B$ contains all points of the 2-dimensional integer lattice that are contained inside a grid of size $\sqrt{n}\times \sqrt{n}$, and have coordinates that are integer multiples of $t$.
Let 
\[
A = \left(\mathbb{Z}_{3\sqrt{n}} \times \mathbb{Z}_{3\sqrt{n}}\right) \setminus B.
\]
Noting that $|A|=|T|$, let $f_0:T \to X$ be an arbitrary bijection.

We next define the edge weights of $G$.
Let $g:V\setminus T \to B$ be an arbitrary bijection.
We define a matching $\mu$ on the set of non-terminals as follows.
Consider some non-terminal vertex $v\in V\setminus T$, such that $g(v) = (it,2jt)$, for some integers $i,j$.
Let $v'\in V \setminus T$ be such that 
$g(v') = (it, (2j+1)t)$.
The vertex $v$ is matched with $v'$, and we write
$\mu(v)=v'$, and $\mu(v')=v$.
We set $w_{vv'} = 1$.

We define a mapping $\sigma:V\setminus T \to T$ as follows.
For every non-terminal $u \in V\setminus T$, with $g(u)=(it,jt)$, let 
\[
\sigma(u) = f_0^{-1}(it,jt+1).
\]
We say that $\sigma(u)$ is the \emph{anchor} of $u$.
We set $w_{u\sigma(u)} = \alpha$, for some parameter $\alpha>0$ to be specified later.
Finally, we set all other edge weights to be $0$.
This completes the definition of the instance of \dAPplus.

\textbf{The fractional solution:}
We now obtain a fractional solution $d$ with small cost.
For every $x,y\in T$, we set $d(x,y) = \|f_0(x)-f_0(y)\|_2$.
For every $u\in V \setminus T$, we set 
\[
d(u, \sigma(u)) = t,
\]
and
\[
d(u,\mu(u)) = 1.
\]
For every terminal $v\in T$, and non-terminal $u\in V \setminus T$, we set
\[
d(v,u) = d(v, \sigma(u)) + d(\sigma(u), u).
\]
Finally, for every $u\neq u'\in V \setminus T$, with $\mu(u) \neq \mu(u')$, we set
\[
d(u,u') = d(u, \sigma(u)) + d(\sigma(u), \sigma(u')) + d(u', \sigma(u')).
\]

\textbf{Metric constraints:}
We first argue that $d$ is indeed a metric.
It is immediate that $d$ is symmetric, so it suffices to show that it satisfies the triangle inequality.
To that end, it is enough to show that there exists a graph $G'=(V,E')$ with shortest-path metric $d$.
The graph $G'$ contains a clique on $T$, where the length of every edge $\{x,y\}$ is set to $d(x,y)$.
We also add for every $v\in V \setminus T$, an edge $\{v,\sigma(v)\}$ in $E'$ with length $t$, and an edge $\{v,\mu(v)\}$ with length $1$.
One can now check that the shortest-path metric of the resulting graph $G'$ is precisely $d$, which establishes that $d$ is a metric.

\textbf{Spreading constraints:}
We now show that $d$ satisfies the spreading constraints.
Let $S\subset V$, $u\in V$.
Assume first that $u\in T$.
For every $v\in S\setminus T$, we extent $f_0$ to $v$ as follows.
Let $v' = \sigma(v) \in T$ be the anchor of $v$.
Suppose that $f_0(v')=(i,j)$.
We set $f_0(v) = (i,j+1)$.
It is immediate to check that for any $z\in T$, $\|f_0(z)-f_0(v)\|_2 \leq d(z,v)$.
After extending in ths same way $f_0$ to $S$, we have
\begin{align*}
\sum_{v\in S} d(u,v) &\geq \sum_{v\in S} \|f_0(u)-f_0(v)\|_2 \geq (|S|-1)^{1+1/2}/4,
\end{align*}
where the last inequality follows by the fact that $f_0(S)\subset \mathbb{Z}^2$.

It remains to consider the case $u\in S\setminus T$.
We again extend $f_0$ to $S$ as above.
For all $v\in S \setminus \{\mu(u)\}$, we have $\|f_0(u)-f_0(v)\|_2 \leq d(u,v)$.
Therefore, if $\mu(u) \notin S$, then, arguing as above, we are done.
However, if $\mu(u)\in S$, then $\|f_0(u)-f_0(\mu(u))\|_2 = d(u,\mu(u)) + t - 1$.
If $S = \{u, \mu(u)\}$, then the spreading constrain is trivially satisfied.
Therefore, we may assume that there exists $w\in S \setminus \{u,\mu(u)\}$.
It follows by construction that $\|f_0(u)-f_0(w)\|_2 \leq d(u,w) - t + 1$.
Putting everything together, we deduce
\begin{align*}
\sum_{v\in S} d(u,v) &= d(u,w) + d(u,\mu(u)) + \sum_{v\in S \setminus\{w,\mu(u)\}} d(u,v)\\
 &\geq \sum_{v\in S} \|f_0(u)-f_0(v)\|_2 \geq (|S|-1)^{1+1/2}/4,
\end{align*}
as required.
This establishes that the spreading constraints are satisfied.

\textbf{Fractional cost:}
In the fractional solution $d$ we pay $1$ for every pair of matched non-terminals, and we pay $t \cdot \alpha$ for every edge between a terminal and its anchor.
Therefore, the total fractional cost is 
\[
c_{\mathsf{LP}} = O((1+t \cdot \alpha) \cdot n/t^2).
\]

\textbf{Integral cost:}
Consider now an integral solution $f:V \to \mathbb{Z}^2$, extending $f_0$. Suppose first that at least half of the non-terminals are mapped outside the grid $\mathbb{Z}_{3\sqrt{n}} \times \mathbb{Z}_{3\sqrt{n}}$. 
Let $X\subseteq U \setminus T$ be the set of these non-terminals.
It follows that for every $v\in U$, we have $\|f(v)-f(\sigma(v))\|_2 = \Omega(\sqrt{n})$.
Therefore, in this case the cost of the integral solution is at least $\Omega( \alpha \cdot n^{1/2} \cdot n/t^2)$.

Next, suppose that at least half of the non-terminals are mapped inside the grid. Then, for at least a constant fraction of the vertices $v\in U$, we have $\|f(v)-f(\mu(v))\|_2 \geq t$. Therefore, in this case the cost of the integral solution is at least $\Omega(t \cdot n/t^2)$.

We conclude that the cost of an optimal integral solution satisfies
\[
c_{\mathsf{OPT}} = \Omega( (a\cdot n^{1/2} + t) \cdot n/t^2).
\]
Setting $t=n^{1/4}$, and $\alpha=n^{-1/4}$, we obtain that the integrality gap is at least $c_{\mathsf{OPT}} / c_{\mathsf{LP}} = \Omega(k^{1/4})$, concluding the proof.
\end{proof}

The above gap construction generalizes to arbitrary fixed dimension $d\geq 1$, giving the following result.

\begin{theorem}
For any fixed $d\geq 1$, the integrality gap of the spreading LP for the unbounded version of \dAPplus~is $\Omega(k^{1/(2d)})$.
\end{theorem}

\section{Hardness}
\label{sec:hardness}

\subsection{Bounded \dAPplus}
\label{sec:bounded-hardness}
Recall that in the bounded case, we need to output a map $f: V \to
\Z_{\sqrt{n}} \times \Z_{\sqrt{n}}$. The hardness for this problem is
the following:

\begin{theorem}
  \label{thm:bdd-hard}
  For any constant $\eps > 0$, the bounded \twoAPplus~problem is NP-hard
  to approximate better than a factor of $\Omega(k^{1/2 - \eps})$.
\end{theorem}

\begin{proof}
  We reduce from the $3$Partition problem: an instance consists of $3n$
  integers $a_1, a_2, \ldots, a_{3n}$, such that each number $a_i \in
  (B/4, B/2)$, and $\sum_i a_i = n\cdot B$. The goal is to decide if
  there exists a partition of the $a_i$s into $n$ triples, each of which
  sum to $B$. The problem is known to be strongly NP-hard even when $B$
  is polynomial in $n$; say $B = n^c$ for some constant $c$. (We can
  assume that $c$ is a large constant.)

  \begin{figure}[ht]
   \centering
   \includegraphics[scale=0.5]{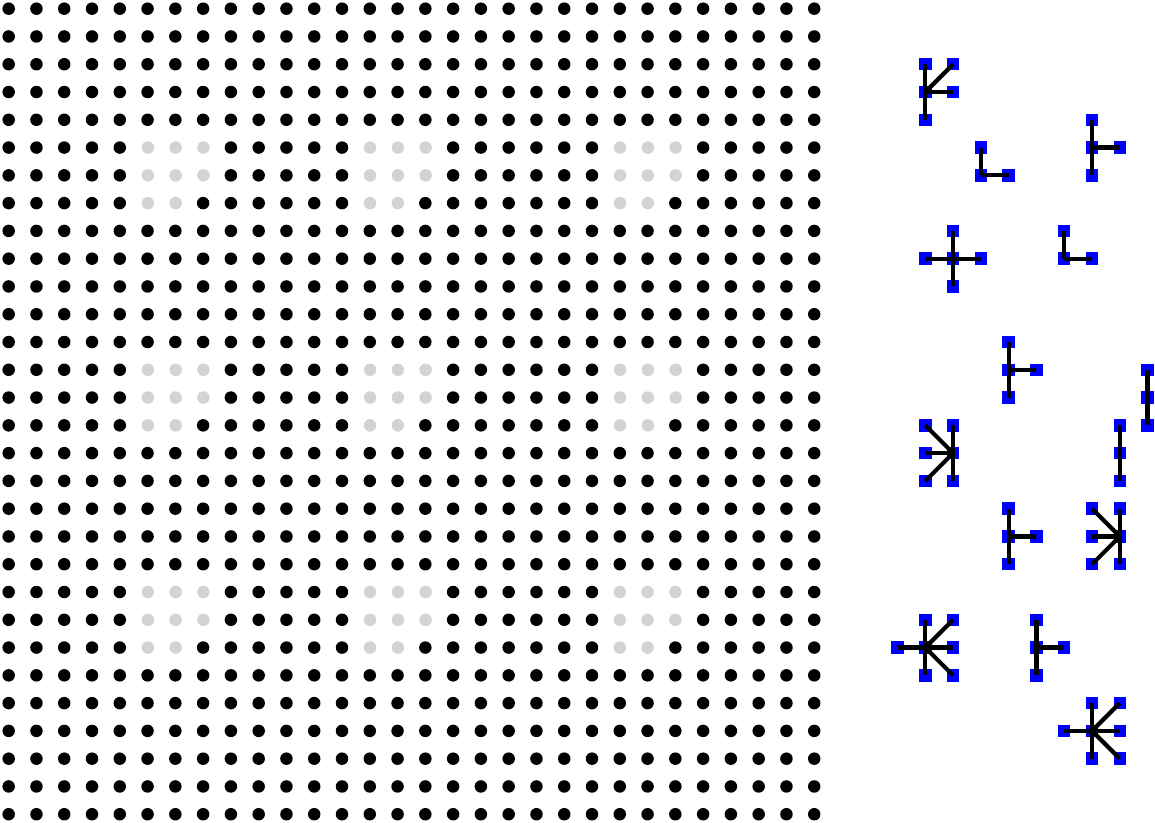}
   \caption{Hardness instance with nine holes.}
   \label{fig:hard1}
  \end{figure}

  Consider a graph $G$ with $3n+1$ connected components. The first $3n$
  components are stars $S_1, S_2, \ldots, S_{3n}$ of size $a_1, a_2,
  \ldots, a_{3n}$; these edges within the stars have unit weight. The
  last component has size $N^2 - nB$, where $N = n^{c'}$ for $c' \gg c$;
  it does not matter what the edges within this component are, since
  they have weight zero. All vertices in this component are terminals,
  which have to be mapped onto a region of size $N \times N$ in the form
  of a grid, from which some vertices have been deleted. In particular,
  choose a ``net'' of $n$ points from this $N \times N$ grid which are
  all at distance at least $n^{c' - 1}$ from each other, and also from
  the boundary of the grid.  Around each of these net points, remove
  exactly $B$ vertices from among the closest vertices. All vertices of
  the stars $\{C_i\}_i$ must map into these holes. All the star edges
  have unit weight, all the edges in the grid have zero weight.
  Figure~\ref{fig:hard1} shows an example where the
  grid is fixed, with nine holes, and there are a collection of stars
  which have to be mapped into these holes.

  If there is indeed a valid $3$Partition, there are $\sum_i (a_i-1) =
  n(B-3)$ edges in total in all these stars, each of which is stretched
  to at most $O(\sqrt{B})$, giving a total cost of $O(nB^{3/2}) = O(n^{3c/2 +
    1})$. If there is no valid $3$Partition, at least one edge must be
  stretched to a distance of at least $n^{c' - 1} - O(\sqrt{B}) =
  \Omega(n^{c' - 1})$. If we set $c' \geq 2c/\eps$, then the gap is
  $N^{1-\eps}$ versus $N^{\eps}$.  Using the fact that $k = \Theta(N^2)$
  gives us the claimed hardness.
\end{proof}

The above hardness reduction directly generalizes to arbitrary fixed dimension, giving the following.

\begin{theorem}
  For any constant $d\geq 1$, and for any constant $\eps > 0$, the bounded \dAPplus~problem is NP-hard
  to approximate better than a factor of $\Omega(k^{1/d - \eps})$.
\end{theorem}

\subsection{Unbounded \dAPplus}
\label{sec:unbounded-hardness}

Now to the unbounded case, where we need to merely output a map $f: V
\to \Z \times \Z$. The hardness for this problem is the following:

\begin{theorem}
  \label{thm:unbdd-hard}
  For every $\eps > 0$, the unbounded \twoAPplus~problem is NP-hard to
  approximate better than a factor of $\Omega(k^{1/4 - \eps})$.
\end{theorem}

\begin{proof}
  We perform the same reduction as in Theorem~\ref{thm:bdd-hard}, but we
  now ``pad'' the $N \times N$ grid to get an even larger grid of
  dimensions $N' \times N'$, with $N' = N^2\cdot B^2$. Each of the
  vertices in each of the stars is now connected to the center of the
  grid with an edge of weight $w$. Recall that the edges of the stars
  all have unit weight. All the vertices in the larger grid are
  terminals. 
  We set $w := B^{1/2}/N$.


  Because of the padding, every mapping that maps the stars into the
  holes is better than any mapping that maps even a single vertex to
  outside the padded grid. Indeed, if we were to map even one vertex of
  one of the stars to outside the big grid (i.e., it is an ``illegal''
  map), when is mapped, we would incur a cost of $w \cdot N'/2 =
  \Omega(N B^{5/2})$ due to the single associated $w$-weight edge. On
  the other hand, the worst possible cost incurred by \emph{any} legal
  map that places all the $nB$ star nodes into the holes is $N\cdot nB +
  w \cdot nB \cdot N = O(nNB)$, which is much smaller --- hence we would
  never consider illegal maps over legal ones.

  Hence it suffices to consider only maps where all vertices are mapped
  into the holes. Since we have an extra cost due to the $w$-weight
  edges, we need to account for this. The maximum cost due to the
  $w$-weight edges for any such ``legal'' map is at most $w\cdot (\sum_i
  a_i) \cdot N/2$, since there are $\sum_i a_i$ nodes, and the holes are
  at distance at most $N/2$ from the center of the grid. This gives us
  $O(w\cdot nB\cdot N)$. Using $w := B^{1/2}/N$, the additional cost due
  to the $w$-edges in this ``legal'' case would be $O(nB^{3/2})$, which
  would not interfere with the hardness reduction of the previous
  section, except changing some constants.

  This allows us to inherit the gap from Theorem~\ref{thm:bdd-hard},
  which is $N^{\eps}$ versus $N^{1-\eps}$, but the number of terminals
  in this padded instance is now $k \approx (NB)^4 = N^{4(1+\eps)}$.
  Hence the hardness is $\Omega(k^{1/4 - \eps'})$, as claimed.
\end{proof}

The above hardness directly generalizes to arbitrary fixed dimension, giving the following.

\begin{theorem}
  For any constant $d\geq 1$, and for any constant $\eps > 0$, the unbounded \dAPplus~problem is NP-hard to
  approximate better than a factor of $\Omega(k^{1/(2d) - \eps})$.
\end{theorem}






\section{An Algorithm for (Slightly Violated) Bounded \dAPplus}
\label{sec:violation}

Our algorithm for unbounded \twoAPplus na\"{\i}vely embeds the points
into a region of size at most $\sqrt{kn} \times \sqrt{kn}$, and gives an
$O(\sqrt{k} \log n)$-approximation. We can trade off the two, by giving
an embedding into a region of size at most $\sqrt{(1+\e)n} \times
\sqrt{(1+\e)n}$, and gives an $O(\sqrt{k/\eps} \log n)$-approximation.
We now sketch this embedding.

First, assume $\e \leq \frac12$. Define size thresholds $s_i := \e
\frac{n}{k} (1+\e)^{i}$ for $i \geq 0$. Define size ``buckets''
$\{B_i\}_{i = 1}^{L}$ for $L = O(\log k)$, where $B_0 := (0, s_0]$, and
\[ B_i := ( s_{i-1}, s_i ] . \] As in Section~\ref{sec:rounding}, the
$0$-extension gives us a set of clusters $C_1, C_2, \ldots, C_k$. Let
$k_i$ be the number of these clusters whose sizes fall in the bucket
$B_i$, where $\sum_i k_i = k$. Moreover,
\begin{gather}
  k_0 \cdot s_0 + \sum_{i \geq 1} k_i \cdot s_i \leq k \cdot \e
  \frac{n}{k}  + (1 + \e)n \leq (1+2\e)n, 
\end{gather}
where in the first inequality, the first term follows from $k_0 \leq k$
and the second from the fact that all $k_i$ clusters have sizes in
$(s_{i-1}, s_i]$. An immediate corollary is that
\begin{gather}
  \sum_{i \geq 0} k_i \cdot \ceil{(1 + \e)^i} \leq \sum_{i \geq 0} k_i
  \cdot (1 + (1 + \e)^i) \leq \bigg( 1 + \frac{(1+2\e)}{\e} \bigg) k
  = \frac{(1+3\e)}{\e} k. \label{eq:4}
\end{gather}

The algorithm is now simple: we break up the $\sqrt{(1+4\e)n} \times
\sqrt{(1+4\e)n}$ grid into about $\eps \frac{n}{k}$ sub-squares, each of
size at least $\frac{(1+3\e)}{\e} k$. (The difference between $(1 +
4\e)$ in the size of the grid and $(1+3\e)$ in the size of sub-squares
is to account for the fact that some of the numbers are not perfect
squares; for simplicity of exposition we gloss over these details.)

The points of each such sub-square are labeled, with $k_i \cdot
\floor{(1 + \e)^i}$ of them being labeled~$i$. By~(\ref{eq:4}) we know
that there are enough points and hence such a labeling is possible. Now
for each of the $k_i$ clusters whose size falls in bucket $B_i$, we will
map $\ceil{(1+\e)^i}$ of that cluster's points into each sub-square.
(This is in contrast to the embedding from Section~\ref{sec:rounding}
that creates $n/k$ sub-squares each of size about $k$, and maps one
point from each cluster into each sub-square.) Since there are $\eps
\frac{n}{k}$ sub-squares, and each of them can take $\ceil{(1+\e)^i}$
points from this cluster, we would have allocated enough room for the at
most $s_i = \eps \frac{n}{k} (1+\e)^i$ points in this cluster. Finally,
we use the same arguments as Section~\ref{sec:rounding} to embed each of
the clusters onto the points allocated to it. The total stretch is
$O(\sqrt{k/\e})$, over and above the stretch due to the $0$-extension
and Rotter-Vygen.

The above argument directly generalizes to any fixed dimension $d\geq 1$, giving the following result.

\begin{theorem}
For any fixed $d\geq 1$, and for any $\eps>0$, there exists a polynomial-time, $O((k / \eps)^{1/d} \log n)$-approximation algorithm for unbounded \dAPplus, such that the image of the embedding is contained in $\{1,\ldots, ((1+\eps) n)^{1/d}\}^d$.
\end{theorem}


\end{document}